\documentclass[runningheads]{llncs}
\usepackage[T1]{fontenc}
\usepackage{amsmath,amssymb}
\usepackage{refcount}
\usepackage{graphicx}
\usepackage{cite}
\usepackage{algorithmic}
\usepackage{booktabs}
\usepackage{textcomp}
\usepackage{xcolor}
\usepackage{mathpartir}
\usepackage{mathtools}
\usepackage{hyperref}
\usepackage{stmaryrd}
\usepackage{subcaption}
\usepackage{booktabs}
\usepackage{tikz}
\usepackage{array}
\usepackage[rightcaption]{sidecap}
\usepackage{orcidlink}
\usepackage[misc]{ifsym}
\usepackage{chngcntr}
\usepackage{apptools}

\AtAppendix{\counterwithin{lemma}{section}}
\AtAppendix{\counterwithin{theorem}{section}}

\begin{document}
\title{Introducing SWIRL: An Intermediate Representation Language for Scientific Workflows}
\titlerunning{An Intermediate Representation Language for Scientific Workflows}
%
\author{Iacopo Colonnelli\inst{1}
\orcidlink{0000-0001-9290-2017} 
\and
Doriana Medić(\Letter)\inst{1}
\orcidlink{0000-0002-7163-5375} 
\and
Alberto Mulone\inst{1}
\orcidlink{0009-0009-2600-613X}
\and
Viviana Bono\inst{1}
\orcidlink{0000-0002-2533-0511}
\and
Luca Padovani\inst{2}
\orcidlink{0000-0001-9097-1297}
\and
Marco Aldinucci\inst{1}
\orcidlink{0000-0001-8788-0829}
}
\authorrunning{Colonnelli, Medić, Mulone, Bono, Padovani and Aldinucci}
%
\institute{University of Turin, Italy  
\email{\{name,surname\}@unito.it}
\and
University of Camerino, Italy 
\email{luca.padovani@unicam.it}}

%
\maketitle              
\begin{abstract}
  In the ever-evolving landscape of scientific computing, properly supporting the modularity and complexity of modern scientific applications requires new approaches to workflow execution, like seamless interoperability between different workflow systems, distributed-by-design workflow models, and automatic optimisation of data movements. In order to address this need, this article introduces SWIRL, an intermediate representation language for scientific workflows. In contrast with other product-agnostic workflow languages, SWIRL is not designed for human interaction but to serve as a low-level compilation target for distributed workflow execution plans. 
  The main advantages of SWIRL semantics are low-level primitives based on the send/receive programming model and a formal framework ensuring the consistency of the semantics and the specification of translating workflow models represented by Directed Acyclic Graphs (DAGs) into SWIRL workflow descriptions. Additionally, SWIRL offers rewriting rules designed to optimise execution traces, accompanied by corresponding equivalence. 
  An open-source SWIRL compiler toolchain has been developed using the ANTLR Python3 bindings.
  
\end{abstract}

\newenvironment{thref}[1]{%
	\vspace*{0.3cm}
	\noindent {\bf  Theorem #1.}}

\newenvironment{lmref}[1]{%
	\vspace*{0.3cm}
	\noindent {\bf  Lemma #1.}}

\newcommand{\dori}[1]{\textcolor{orange}{Dori: #1}}


\newcommand{\dom}{\mathrm{dom}}
\newcommand{\domplus}{\mathrm{dom}^{+}}
\newcommand{\In}{\mathit{In}}
\newcommand{\Out}{\mathit{Out}}
\newcommand{\Transfer}{\mathcal{T}}
\newcommand{\Neigh}{\mathcal{N}}
\newcommand{\Ports}{\mathcal{P}}
\newcommand{\Data}[1]{D_{#1}}
\newcommand{\mapping}{\mathcal{M}}
\newcommand{\dip}{\mathcal{D}}
\newcommand{\parop}{\mathbin{|}}
\newcommand{\Map}{\mathit{Map}}
\newcommand{\Scatter}{\mathit{Scatter}}
\newcommand{\Gather}{\mathit{Gather}}
\newcommand{\driver}{\mathit{driver}}
\newcommand{\config}[2]{\left\langle #1, #2\right\rangle}
\newcommand{\lconfig}[3]{\left\langle #1, #2, #3\right\rangle}
\newcommand{\sconfig}[4]{\left\langle #1, #2, #3, #4\right\rangle}
\newcommand{\ex}{e}
\newcommand{\target}[1]{\left\langle #1\right\rangle}
\newcommand{\activestate}{A}
\newcommand{\inactivestate}{I}
\newcommand{\anystate}{St}
\newcommand{\commpath}[2]{\mathcal{C}(#1, #2)}
\newcommand{\deppath}[2]{\mathcal{D}\left(#1, #2\right)}
\newcommand{\lts}[1]{\xrightarrow{\; #1 \;}}
\newcommand{\dep}{\mathtt{deploy}}
\newcommand{\udep}{\mathtt{undeploy}}
\newcommand{\nul}{\mathbf{0}}
\newcommand{\drop}[1]{\mathtt{drop}(#1)}
\newcommand{\exec}[1]{\mathtt{exec}(#1)}
\newcommand{\send}[1]{\mathtt{send}(#1)}
\newcommand{\move}[1]{\mathtt{move}(#1)}
\newcommand{\trans}[1]{\mathtt{transfer}(#1)}
\newcommand{\flow}[1]{\xmapsto[]{#1}}
\newcommand{\recv}[1]{\mathtt{recv}(#1)}
\newcommand{\step}{\mathtt{s}}
\newcommand{\tr}{\mathtt{t}}
\newcommand{\para}{\;|\;}
\newcommand{\comm}[1]{\textcolor{blue}{#1}}
\newcommand{\commred}[1]{\textcolor{red}{#1}}
\newcommand{\conf}{\mathtt{C}}
\newcommand{\work}{\mathtt{W}}
\newcommand{\system}{\mathtt{M}}
\newcommand{\hwork}{\mathtt{H}}
\newcommand{\rel}{\mathcal{R}}
\newcommand{\eq}{\sim}
\newcommand{\causal}{\sqsubseteq}
\newcommand{\conc}{\smile}
\newcommand{\loc}{\mathtt{loc}}
\newcommand{\enc}[1]{\llbracket{#1}\rrbracket}
\newcommand{\encod}[1]{\llparenthesis{#1}\rrparenthesis}
\newcommand{\pathset}{{\bf{C}}}
\newcommand{\conflict}{\bowtie}
\newcommand{\messa}[2]{{#1}\rightarrowtail {#2}}
\newcommand{\congr}{\equiv}
\newcommand{\load}{Q}
\newcommand{\inout}[2]{{\if#1\emptyset\emptyset\else\{#1\}\fi}\mapsto {\if#2\emptyset\emptyset\else\{#2\}\fi}}
\newcommand{\localopt}[1]{\llbracket {#1}\rrbracket}
\newcommand{\traceopt}[1]{\lbag {#1}\rbag}
\newcommand{\execb}[1]{{\bf{\mathtt{\bf{exec}}\bf{(#1)}}}}
\newcommand{\sendb}[1]{{\bf{\mathtt{\bf{send}}\bf{(#1)}}}}
\newcommand{\recvb}[1]{{\bf{\mathtt{\bf{recv}}\bf{(#1)}}}}
\newcommand{\relabel}[1]{[{#1}]}
\newcommand{\barb}[1]{\downarrow_{#1}}
\newcommand{\wbarb}[1]{\Downarrow_{#1}}
\newcommand{\bisim}{\approx}
\newcommand{\opt}{\mathtt{O}}
\newcommand{\weak}{\Rightarrow}
\newcommand{\prefset}{A}
\newcommand{\optset}{\mathcal{O}}
\newcommand{\signal}{\mathtt{0}}
\newcommand{\instance}{\mathtt{I}}
\newcommand{\mem}{M}
\newcommand{\mapdati}{\mathcal{I}}
\newcommand{\operator}{\mathtt{o}}

\newcommand{\inoutset}[2]{{#1}\mapsto {#2}}
\newcommand{\datadis}{G}
\newcommand{\eps}{\epsilon}
\newcommand{\marking}{\Phi}
\newcommand{\datalocation}{A}
\newcommand{\doriana}[1]{\textcolor{red}{#1}}

\begin{section}{Introduction}
  Workflows have been widely used to model large-scale scientific workloads. The explicit definition of true dependencies between subsequent steps allows inferring concurrent execution strategies automatically, improving performances, and transferring input and output data wherever needed, fostering large-scale distributed executions. However, current Workflow Management Systems (WMSs) struggle to keep up with the ever-more demanding requirements of modern scientific applications, such as \emph{interoperability} between different systems, \emph{distributed-by-design} workflow models, and \emph{automatic optimisation of data movements}.

  With the advent of BigData, adopting a proper \emph{data management} strategy has become a crucial aspect of large-scale workflow orchestration. Avoiding unnecessary data movements and coalescing data transfers are two established techniques for performance optimisation in distributed executions. Moving computation near data to remove the need for data transfers is the underlying principle of several modern approaches to large-scale executions, like Resilient Distributed Datasets \cite{Zaharia:2012} and in-situ workflows \cite{Ayachit:16}. 
  
  WMSs' \emph{interoperability} is an open problem in scientific workflows, which hinders reusability and composability. Despite several attempts to model product-agnostic workflow languages \cite{CWL:2022} and representations \cite{Plankensteiner:2011} present in the literature, these solutions capture only a subset of features, forcing WMSs to reduce their expressive power in the name of portability. The main issue in unifying workflow representations resides in the heterogeneity of different WMSs' APIs and programming models tailored to the needs of a domain experts. Conversely, moving the interoperability efforts to the lower level of the workflow execution plan representation is a promising but still relatively unexplored alternative.
  
  The \emph{heterogeneity} in contemporary hardware resources and their features, further exacerbated by the end-to-end co-design approach \cite{Reed:22}, requires WMSs to support a large ecosystem of execution environments (from HPC to cloud, to the Edge), optimisation policies (performance vs. energy efficiency) and computational models (from classical to quantum). However, maintaining optimised executors for such diverse execution targets is an overarching effort. In this setting, a just-in-time compilation of target-specific execution bundles, optimised for a single workflow running in a single execution environment, would be a game-changing approach. Indeed, this approach allows for the efficient use of resources, as the compilation is done at the time of execution, taking into account the specific characteristics of the execution environment. It also ensures the effectiveness of the execution, as the compiled bundle is optimised for the specific workflow, leading to improved performance.
  
  This work presents SWIRL, a ``Scientific Workflow Intermediate Representation Language''. Unlike other product-agnostic workflow languages, SWIRL is not intended for human interaction but serves as a low-level compilation target for distributed workflow execution plans. It models the execution plan of a location-aware workflow graph as a distributed system with send/receive communication primitives. This work provides a formal method to encode a workflow instance into a distributed execution plan using these primitives, promoting interoperability and composability of different workflow models. It also includes a set of rewriting rules for automatic optimisation of data communications with correctness and consistency guarantees. The optimised SWIRL representation can then be compiled into one or more self-contained executable bundles, making it adaptable to specific execution environments and embracing heterogeneity. 

  The SWIRL implementation follows the same line as the theoretical approach, separating scientific workflows' design and runtime phases. A SWIRL-based compiler translates a workflow system $\work$ to a high-performance, self-contained workflow execution bundle based on send/receive communication protocols and runtime libraries, which can easily be included in a Research Object \cite{Bechhofer13}, significantly improving reproducibility.

  In detail, Sec.~\ref{sc:background} introduces a generic formalism for representing distributed scientific workflow models, while the related work and the comparison with the SWIRL language is given in Sec.~\ref{sc:related-work}. Sec.~\ref{sec:methodology} introduces the SWIRL semantics, and Sec.~\ref{sec:optimisation} derives the rewriting rules used for optimisation. Sec.~\ref{sc:implementation} describes the implementation of the SWIRL compiler toolchain while Sec.~\ref{sec:evaluation} shows how to model the 1000 Genomes workflow \cite{Silva:19}, a Bioinformatics pipeline aiming at fetching, parsing and analysing data from the 1000 Genomes Project \cite{1000Genomes:15} into SWIRL system. 
  Finally, Sec.~\ref{sc:conclusion} concludes the article. Full proofs and additional material can be found in Appendix~\ref{app:a} while the experiment is in~\cite{zenodo-swirl}.
\end{section}

\begin{section}{Background and related work}\label{sc:background}
  This section gathers the related work (Sec.~\ref{sc:related-work}) and introduces a formal representation of scientific workflows (Sec.~\ref{subsec:workflows}) and their mapping onto distributed and heterogeneous execution environments (Sec.~\ref{subsec:distributed-workflows}). 

  \begin{subsection}{Related work} \label{sc:related-work}
    \emph{Location-aware WMSs.}
    Grid-native WMSs typically support distributed workflows out of the box, providing automatic scheduling and data transfer management across multiple execution locations. However, all the orchestration aspects are delegated to external, grid-specific technologies, limiting the spectrum of supported execution environments. For instance, Triana \cite{Taylor:2007}, Askalon \cite{Askalon:2007} and Pegasus \cite{pegasus:2019} delegate tasks offloading and data transfers to the GAP interface \cite{Taylor:2003}, the GLARE library \cite{Siddiqui:2005}, and HTCondor \cite{HTCondor:2005}, respectively. 
    
    Recently, a new class of location-aware WMSs is bringing advantages in performance and costs of workflow executions on top of heterogeneous distributed environments. StreamFlow \cite{20:streamflow:tetc} allows users to explicitly map each step onto one or more locations in charge of its execution. It relies on a set of connectors to support several execution environments, from HPC queue managers to microservices orchestrators. Jupyter Workflow \cite{21:FGCS:jupyflow} transforms a sequential computational notebook into a distributed workflow by mapping each cell into one or more execution locations, semi-automatically extracting inter-cell data dependencies from the code, and delegating the runtime orchestration to StreamFlow. Mashup \cite{mashup:22} automatically maps each workflow step onto the best-suited location, choosing between traditional Cloud VMs and serverless platforms.

    Each tool has its own strategy to derive an execution plan from a workflow graph without relying on an explicit and consolidated intermediate representation. Moreover, none of them formalise this derivation process, hiding its details inside the WMS's codebase. Instead, relying on a common intermediate language like SWIRL would allow interoperability between different tools and formal correctness guarantees on the adopted optimisation strategies.  

    \emph{Formal models for distributed workflows}
    In the literature, the number of different WMSs is notable ~\cite{numberwms}, however, up to our knowledge, there are only a few WMS for which formal models have been developed: Taverna~\cite{TuriMGRO07}, employing the lambda calculus~\cite{Moggi89} to define the workflow language in functional terms; Kepler~\cite{Kepler:2006} adopting Process Networks~\cite{KahnM77} and BPEL~\cite{OuyangVABDH07}, where the workflow language is formalised with Petri Nets. YAWL~\cite{AalstH05} is another workflow language based on Petri Nets extended with constructs to address the multiple instances, advanced synchronisation, and cancellation patterns. It provides a detailed representation of workflow patterns~\cite{Aalst:2003} supported by an open-source environment.

    Process algebra, in particular, different versions of $\pi$-calculus~\cite{picalculus} are suited to model the workflow system due to the ability of processes to change their structure dynamically. A class of workflow patterns has been precisely defined using the execution semantics of $\pi$-calculus, in~\cite{PuhlmannW05}, while the basic control flow constructs modelled by $\pi$-calculus are given in~\cite{workflowpimed}.
    A distributed extension of $\pi$-calculus~\cite{Hennessy:2007} is examined as a formalisation for distributed workflow systems in~\cite{MedicA23}, providing a discussion on the flexibility of the proposed representation. Aside from $\pi$-calculus, CCS (Calculus of Communication Systems)~\cite{ccs} models Web Service Choreography Interface descriptions. 
  \end{subsection}

  \begin{subsection}{Scientific workflow models}\label{subsec:workflows}
    A generic workflow can be represented as a \emph{directed bipartite graph}, where the nodes refer to either the computational \emph{steps} of a modular application or the \emph{ports} through which they communicate, and the edges encode \emph{dependency relations} between steps.
    
    \begin{definition}\label{def:workflow}
      A workflow is a directed bipartite graph $W = (S, P, \dip)$ where $S$ is the set of steps, $P$ is the set of ports, and $\dip \subseteq (S \times P) \cup (P \times S)$ is the set of dependency links.
    \end{definition}

    In the considered graph, one port can have multiple output edges meaning that more steps are dependent on it. The sets of input/output ports (steps) of a step (port) are defined with the following definition. 

    \begin{definition}
      Given a workflow $W=(S,P,\dip)$,a step $s\in S$ and a port $p\in P$,
      the sets of input and output ports of $s$ are denoted with $\In(s)$ and $\Out(s)$, respectively, and defined as:
      {\small
        \begin{equation*}
          \In(s)=\{p'\para (p',s)\in \dip\} \qquad \Out(s)=\{p'\para (s,p')\in \dip\}
        \end{equation*}
      }
      while the sets of input and output steps of $p$ are denoted with $\In(p)$ and $\Out(p)$, respectively, and defined as:
      {\small
        \begin{equation*}
          \In(p)=\{s'\para (s',p)\in \dip\} \qquad \Out(p)=\{s'\para (p,s')\in \dip\}
        \end{equation*}
      }
    \end{definition}

    Traditionally, scientific workflows are modelled using a dataflow approach, i.e., 
    following \emph{token-pushing} semantics in which tokens carry \emph{data values}.
    The step executions are enabled by the presence of tokens in their input ports and produce new tokens in their output ports. In general, a single workflow model can generate infinite \emph{workflow instances}. Different instances preserve the same graph structure but differ in the values carried by each token.
    
    \begin{definition}\label{def:datainstance}
      A workflow instance is a tuple $(W, \Data{}, \mapdati)$ where $W=(S,P,\dip)$ is a workflow, $\Data{}$ is a set of data elements, and $\mapdati\subseteq (\Data{}\times P)$ is a mapping relation connecting each data element $d \in \Data{}$ to the port $p \in P$ that contains it.
    \end{definition}

    \begin{definition}
      Given a workflow instance $(W,\Data{},\mapdati)$, where $W=(S,P,\dip)$, and a step $s \in {S}$, the sets of input and output data elements of $s$ are denoted with $\In^{D}(s)$ and $\Out^{D}(s)$, respectively, and defined as:
      {\small
      \begin{align*}
        \In^{D}(s) =\{d\para (d,p)\in \mapdati \;\wedge\; p\in\In(s)\}\qquad
        \Out^{D}(s) =\{d\para (d,p)\in \mapdati \;\wedge\; p\in\Out(s)\}
      \end{align*}}
    \end{definition}

    Introducing more precise evaluation semantics, triggering strategies, or limitations on the dependencies structure can specialise this general definition to an actual workflow model (e.g., a Petri Net \cite{Reisig:1998} or Coloured Petri Nets \cite{Jensen:1989}, a Kahn Processing Network \cite{Kahn:74}, or a Synchronous Dataflow Graph \cite{Lee:87}).  
  \end{subsection}
 
  \begin{subsection}{Distributed workflow models} \label{subsec:distributed-workflows}
    A \emph{distributed workflow} is a workflow whose steps can target different \emph{deployment locations} in charge of executing them. To compute the step, the corresponding location must have access to or store all the input data elements, additionally, it will store all the output data elements on its local scope. Locations can be \emph{heterogeneous}, exposing different hardware devices, software libraries, and security levels. Consequently, the steps are explicitly mapped onto execution locations depending on their computing requests.
    Given that, a distributed workflow model must contain a specification of the workflow structure, the set of available locations, and a mapping relation between steps and locations.
  
    \begin{definition}\label{def:distributed-workflow}
      A distributed workflow is a tuple $\left(W, L, \mapping\right)$, where $W = \left(S, P, \dip\right)$ is a workflow, $L$ is the set of available locations, and $\mapping\subseteq (S \times L)$ is a mapping relation stating which locations are in charge of executing each workflow step.
    \end{definition}

    Each location can execute multiple steps on it, and a single step can be mapped onto multiple locations. 
    Multiple steps related to a single location introduce a \emph{temporal constraint}: all the involved steps compete to acquire the location's resources. They can be serialised if the location does not have enough resources to execute all of them concurrently. Conversely, multiple locations related to a single step express a \emph{spatial constraint}: all involved locations must collaborate to execute the step. 
    This work does not impose any particular strategy for scheduling different step executions on a single location when temporal constraints arise. However, it is helpful to know the \textit{work queue} of a given location $l$, i.e., the set of steps mapped onto it.
    \begin{definition}
      Given a distributed workflow $\left(W, L, \mapping\right)$, where $W = \left(S, P, \dip\right)$, and a location $l\in L$, the set of steps mapped onto $l$ is called the work queue of $l$, denoted as $\load(l)$ and defined as:
        $\load(l) = \{s\para l \in \mapping(s)\} $.
    \end{definition}

    Similarly to what was discussed in Sec.~\ref{subsec:workflows}, a single distributed workflow model can generate potentially infinite \emph{distributed workflow instances} with different data elements and condition evaluations.
    
    \begin{definition}\label{def:distributeddatainstance}
      A distributed workflow instance is a tuple $\instance=\left(W, L, \mapping, \Data{}, \mapdati\right)$ where $(W,L,\mapping)$ is a distributed workflow, $D$ is a set of data elements, and $\mapdati\subseteq (\Data{}\times P)$ is a mapping relation connecting each data element $d \in \Data{}$ to the port $p \in P$ that contains it.
    \end{definition}

    \begin{example} \label{ex:distributed-workflow}
      Fig.~\ref{fig:distributed-workflow-example} shows an example of a distributed workflow model. A step $s_{1}$ produces two different output data elements $d_1$ and $d_2$, which are mapped to ports $p_{1}$ and $p_{2}$. The second and the third step, $s_{2}$ and $s_{3}$ depend on the data elements on the ports $p_{1}$ and $p_{2}$, respectively. None of them produces other outputs. This workflow is mapped onto four locations. Step $s_{1}$ is executed on location $l_{d}$, while $s_{2}$ is offloaded to $l_{1}$ and step $s_{3}$ is mapped to two locations $l_{2}$ and $l_{3}$. Using definitions above, Fig.~\ref{fig:distributed-workflow-example} can be written as follows:
      {\small
      \begin{equation*}\label{eq:distributed-workflow-example}
        \begin{aligned}
          W &= \left(
                  \left\{s_{1}, s_{2}, s_{3}\right\}\!,
                  \left\{p_{1}, p_{2}\right\}\!,
                  \left\{(s_{1}, p_{1}), (s_{1}, p_{2}), (p_{1}, s_{2}), (p_{2}, s_{3})\right\}
                \right)\\
          L &= \left\{l_{d}, l_{1}, l_{2}, l_{3}\right\} \qquad\qquad\qquad
          \mapping = \left\{\left(s_{1}, l_{d}\right)\!, \left(s_{2}, l_{1}\right)\!, \left(s_{3}, l_{2}\right)\!, \left(s_{3}, l_{3}\right)\right\}\\
          \Data{}& = \left\{d_{1}, d_{2}\right\} \qquad\qquad\;\;\qquad\qquad
        \mapdati = \left\{\left(d_{1}, p_{1}\right)\!, \left(d_{2}, p_{2}\right)\right\}
        \end{aligned}
      \end{equation*}}
    \end{example}

    \begin{SCfigure}[][t]
      \centering
      \vspace*{-2mm}
      \includegraphics[scale=0.8]{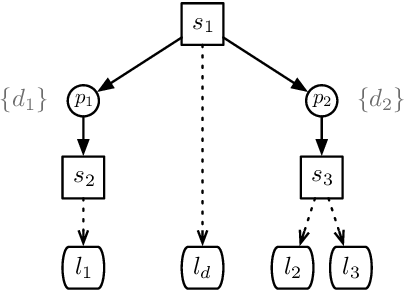}%
      \caption{\small{Example of a distributed workflow model. Steps are represented as squares and ports as circles. Dependency links between steps and ports are depicted as arrows with black-filled heads. Locations are represented as squashed rectangles. Mapping relations are expressed as dotted arrows. A potential instance of this model can be derived by adding data elements, denoted as sets of values, near their related port.}}
      \label{fig:distributed-workflow-example}
    \end{SCfigure}
  \end{subsection}
\end{section}

\begin{section}{The SWIRL representation} \label{sec:methodology}
  This section introduces SWIRL, a ``Scientific Workflow Intermediate Representation Language''. Given a distributed workflow instance (Sec.~\ref{subsec:distributed-workflows}), SWIRL can model a decentralised execution plan, called \emph{workflow system}, by inferring and projecting execution traces on each involved location and specifying inter-location communications using send/receive primitives. 
  The following sections introduce SWIRL syntax and semantics and derive a procedure to formally encode a workflow instance $\instance$ into a SWIRL workflow system $\work$.
  
  SWIRL models a distributed execution plan as a workflow system $\work$, which can be seen as a parallel composition of \emph{location configurations}, tuples $\lconfig{l}{\Data{}}{\ex}$, containing the location name $l$, the set $\Data{}$ of data elements laying on $l$ at a given time, and the execution trace $\ex$ representing the actions to be executed on $l$.
    
  \begin{definition} \label{def:swirl-syntax}
    The syntax of a workflow system $\work$ is defined by the following grammar: 
    \vspace{-2mm}
    {\small
    \begin{align*}
      \work &::= \lconfig{l}{\Data{}}{\ex}\parallel  (\work_1 \para \work_2)\\
      \ex &::= \mu \parallel \ex_1.\ex_2 \parallel (\ex_1\para\ex_2)\parallel \nul\\
      \mu &::=  \exec{s,F(s),\mapping(s)} \parallel \send{\messa{d}{p},l,l'}\parallel
          \recv{p,l,l'}\\
      F(s) &::= \In^{D}(s)\mapsto \Out^{D}(s) 
    \end{align*}}
  \end{definition}
  Each execution trace $\ex$ is constructed from the predicates $\mu$, which can be composed using two operators: the \emph{sequential execution} $\ex_{1}.\ex_{2}$ and the \emph{parallel composition} $\ex_1\para\ex_2$. The $\nul$ symbol represents the empty trace.

  A predicate $\mu$ represents an action to be performed during workflow execution. Predicates $\send{\messa{d}{p},l,l'}$ and $ \recv{p,l,l'}$ allow transferring the data element $d$ over port $p$ from location $l$ to location $l'$. 
  Modelling ports and data separately seams redundant, but we prefer to keep them divided for the future extensions of the framework, as adding the loops.
  The $\exec{s,F(s),\mapping(s)}$ action represents the execution of step $s$. Besides the name of the step, this predicate contains the set $\mapping(s)$ of locations onto which $s$ is mapped and the \emph{dataflow} $F(s)$, i.e., the set $\In^{D}(s)$ of input data needed by $s$ and the set $\Out^{D}(s)$ of output data produced on each $l \in \mapping(s)$ after the execution of $s$.

  \begin{subsection}{Semantics}\label{subsec:swirl-semantics}
    \begin{figure}[t]
      {\scriptsize
        \begin{mathpar}
          \inferrule*[left={\scriptsize{(\textsc{Id$_{\para}$})}}\quad]
          {}
          {\ex\para\nul \congr \ex}
          \and
          \inferrule*[left={\scriptsize{(\textsc{Id$_{.}$})}}\quad]
          {}
          {\nul.\ex \congr \ex\,\wedge\,\ex.\nul \congr \ex}
          \and
            \inferrule*[left={\scriptsize{(\textsc{Comt$_{u}$})}}\quad]
          {}
          {u\para u' \congr u'\para u},\; u\in {e,\work}
           \vspace*{-2mm}
        \end{mathpar}}
        \vspace*{-2mm}
      \caption{{\small{SWIRL structural congruence rules.}}}
      \label{fig:swirl-congruence-rules}
      \vspace{-2mm}
    \end{figure}

    The SWIRL semantics is defined in terms of a \emph{reduction semantics}.

    \begin{definition}
      The SWIRL semantics is defined by the reduction relation $\lts{\;}$ defined as a smallest relation closed under the rules of Figs.~\ref{fig:swirl-congruence-rules} and~\ref{fig:swirl-semantics}.
    \end{definition}

    The structural congruence properties are reported in Fig.~\ref{fig:swirl-congruence-rules}. The commutativity of the parallel composition in location and the execution trace level is defined with rule \textsc{(Comt$_{u}$)}. For both operators, parallel composition and sequential execution, the identity element is $\nul $ (rules \textsc{(Id$_{\para}$)} and \textsc{(Id$_{.}$)}).
    \begin{figure*}[t]
      {\scriptsize
        \begin{mathpar}
          \inferrule*[left={\scriptsize{(\textsc{Exec})}}\quad]
          {\forall l_i \in \mapping(s) \quad\wedge\quad \In^{D}(s)\subseteq \Data{i}}
          {\prod_{i, l_i \in\mapping(s)}\lconfig{l_i}{\Data{i} }{\exec{s,F(s),\mapping(s)}.\ex_i}\lts{} 
          \prod_{i, l_i \in\mapping(s)}\lconfig{l_i}{\Data{i} \cup \Out^{D}(s) }{\ex_i} }
          \and
          \inferrule*[left={\scriptsize{(\textsc{L-Comm})}}\quad]
          {d\in \Data{}}
          {\lconfig{l}{\Data{}}{\send{\messa{d}{p},l,l }.\ex\para \recv{p,l,l }.\ex'} \lts{}
          \lconfig{l}{\Data{}}{\ex\para \ex'} 
          }
          \and
          \inferrule*[left={\scriptsize{(\textsc{Comm})}}\quad]
          {d\in \Data{}}
          {\lconfig{l}{\Data{}}{\send{\messa{d}{p},l,l' }.\ex} \para  \lconfig{l'}{\Data{}'}{\recv{p,l,l' }.\ex'} \lts{}
          \lconfig{l}{\Data{}}{\ex}\para \lconfig{l'}{\Data{}'\cup \{d\}}{\ex'}
          }
          \and
          \inferrule*[left={\scriptsize{(\textsc{L-Par})}}\quad]
          {\lconfig{l}{\Data{}}{\ex_1 }\lts{}\lconfig{l}{D'_{}}{\ex'_1 }}
          {\lconfig{l}{\Data{}}{\ex_1 \para\ex_2}\lts{}\lconfig{l}{D'_{}}{\ex'_1 \para\ex_2}}
          \and
          \inferrule*[left={\scriptsize{(\textsc{Sec})}}\quad]
          {\lconfig{l}{\Data{}}{\ex_1 }\lts{}\lconfig{l}{D'_{}}{\ex'_1 }}
          {\lconfig{l}{\Data{}}{\ex_1.\ex_2}\lts{}\lconfig{l}{D'_{}}{\ex'_1.\ex_2}}
          \and
          \inferrule*[left={\scriptsize{(\textsc{Par})}}\quad]
          {\work_1\lts{}\work'_1 }
          {\work_1\para \work_2 \lts{}\work'_1  \para \work_2}
          \and 
          \inferrule*[left={\scriptsize{(\textsc{Congr})}}\quad]
          {\work'_1\congr\work_1 \lts{}\work_2 \congr \work'_2}
          {\work'_1\lts{}\work'_2 }
        \end{mathpar}
      }
      \caption{{\small{SWIRL reduction semantics rules.}}}
      \label{fig:swirl-semantics}
      \vspace{-2mm}
    \end{figure*}

    The rules of a SWIRL semantics are depicted in Fig.~\ref{fig:swirl-semantics}. 
    The step execution is performed by the (\textsc{Exec}) rule. It collects all the locations $\mapping(s)$ onto which step $s$ is mapped and synchronises the execution action. The data $\Out^{D}(s)$ produced by the step execution are added to the set $D_{i}$ in all executing locations. Rule (\textsc{L-Comm}) describes local communication, while rule (\textsc{Comm}) represents a data transfer between two locations. In the latter case, the involved data element is \emph{copied} to the targeted (receiving) location. Note that communications do not consume the data element on the sending location.
  
    Assuming that configuration $\lconfig{l}{\Data{}}{\ex_1}$ can be computed, rules (\textsc{L-Par}) and (\textsc{Seq}) allow for the execution of the parallel and the sequential composition inside the same location, respectively.
    The execution of the workflow sub-system as a part of a larger system is allowed by the rule (\textsc{Par}). The (\textsc{Congr}) rule allows the application of structural congruence. 

    When a step $s$ is mapped onto multiple locations, each of them must contain an $\mathtt{exec}$ predicate with the set of involved locations. Such predicates introduce \emph{synchronisation points} among different locations, as all involved execution traces must step forward in a single pass. Additionally, each location must have a copy of the input data $\In^{D}(s)$, requiring multiple $\mathtt{send}$ operations for each element $d \in \In^{D}(s)$, and will own a copy of $\Out^{D}(s)$.
    
    \begin{example} \label{ex:distributed-workflow-instance-syntax}
      The behaviour of the distributed workflow instance given in Fig.~\ref{fig:distributed-workflow-example} can be modelled as a workflow system $\work$ with the following syntax:
      {\small
      \vspace{-2mm}
      \begin{align*}
        &\qquad\qquad \work = \lconfig{l_{d}}{\emptyset}{\ex_{d}}\para \prod^{3}_{i=1} \lconfig{l_{i}}{\emptyset}{\ex_{i}}\\[2mm]
        \ex_{d} &=
          \begin{aligned}[t]
            &\exec{s_1,\inout{\emptyset}{d_1,d_2},\{l_{d}\}}.\big(\send{\messa{d_1}{p_1},l_{d},l_1} \para\\
            &\send{\messa{d_2}{p_2},l_{d},l_2}\para \send{\messa{d_2}{p_2},l_{d},l_3}\big)
          \end{aligned}\\
        \ex_{1} &= \recv{p_1,l_{d},l_1}.\exec{s_2,\inout{d_1}{\emptyset},\{l_1\}}\\
        \ex_{2} &= \recv{p_2,l_{d},l_2}.\exec{s_3,\inout{d_2}{\emptyset},\{l_2,l_3\}}\\
        \ex_{3} &= \recv{p_2,l_{d},l_3}.\exec{s_3,\inout{d_2}{\emptyset},\{l_2,l_3\}}
      \end{align*}}
      In the execution trace $\ex_{d}$, step $s_1$ is sending output data $d_{2}$ to both locations $l_2, l_3 \in \mapping(s_{3})$ through the same port $p_2$.
    \end{example}
  \end{subsection}

  \begin{subsection}{Workflow model encoding}\label{subsec:encoding}
    Example~\ref{ex:distributed-workflow-instance-syntax} describes the encoding of a distributed workflow instance $\instance$ into a workflow system $\work$. This section introduces a formal methodology to perform this encoding automatically for any distributed workflow instance.

    In SWIRL, the execution trace $\ex_{l}$ of a location $l \in L$ models the actions required to execute all the steps in its work queue $\load(l)$. In this respect, $\ex_{l}$ can be seen as the parallel composition of \emph{building blocks} $B_{l}(s)$, one for each $s \in \load(l)$. Each building block $B_{l}(s)$ contains the same sequence of actions: 
    $(i)$ receives all the necessary data elements for the step execution in which case it is necessary to determine all input data elements ($\In^{D}(s)$) and for each element to identify the step producing it ($\In(\mapdati (d_i))$) and the locations on which the steps are mapped to ($\mapping(\In(\mapdati (d_i)))$)
    $(ii)$ executes the step $s$; 
    $(iii)$ sends the produced data elements ($\Out^{D}(s)$) to the locations onto which the receiving steps are mapped (one data element can be sent, over the same port, to the different steps/locations, therefore it is necessary to identify the steps data $d_i$ is sent to with $ \Out(\mapdati (d_i))$ and the locations $l_j$ on which each step is deployed).
      
    \begin{definition}\label{def:buildingblock}
      Given a distributed workflow instance $\instance=\left(W, L, \mapping, \Data{}, \mapdati\right)$, a deployment location $l\in L$  
      and a step $s \in S$ s.t. $l\in\mapping(s)$, the building block representing $s$ in $\ex_{l}$ is denoted by $B_{l}(s)$ and defined as:
      \vspace{-1mm}
      {\footnotesize
      \begin{align*}
      B_{l}(s)= & \left ( \prod^{\forall d_i \in \In^{D}(s)} \prod^{\forall l_j \in \mapping(\In(\mapdati (d_i)))}\recv{\mapdati (d_i), l_j,l}\right ).\\[1mm]
      & \exec{s, \inoutset{\In^{D}(s)}{\Out^{D}(s)},\mapping(s)}.\\[1mm]
      &\left ( \prod^{\forall d_i \in \Out^{D}(s)} \prod^{\forall s_k \in \Out(\mapdati (d_i))}  \prod^{\forall l_j \in \mapping(s_k)}    \send{\messa{d_{i}}{\mapdati (d_i)}, l, l_j }\right )
      \end{align*}
      }

    \end{definition}
    
    Def.~\ref{def:buildingblock} introduces the general form of $B_{l}(s)$, which holds for steps connected to both input and output ports. If a step does not consume input data, as in the case of step $s_{1}$ from Example~\ref{ex:distributed-workflow-instance-syntax}, the receiving part of $B_{l}(s)$ is modelled with $\nul$, resulting in $B_{l_{d}}(s_{1})=\nul.\exec{s,\inout{\emptyset}{d_1},\{l_{d}\}}.\send{\messa{d_1}{p_1},l_{d},l_{1}}$. The same applies to steps that do not produce output data.

    The concept of building blocks $B_{l}(s)$ allows for the modular construction of execution traces by processing one pair $(s, l)$ at a time. Intuitively, for each mapping pair step-location $(s,l)$, corresponding  building blocks $B_{l}(s)$ are made and added to the execution trace of the location $l$. Another important information is the \emph{instance data distribution} on the locations, denoted by $\datadis (l) = \{d | d \in \Data{l}\}$. 
    
    \begin{definition} \label{def:graphenc}
      The encoding function $\enc{\cdot}:  \mathcal{W}_{\instance} \lts{} \mathcal{W}_{\work}$, where $\mathcal{W}_{\instance}$ and $\mathcal{W}_{\work}$ are the sets of distributed workflow instances and workflow systems represented in SWIRL, respectively, is inductively defined as follows:
      {\footnotesize
      \vspace{1mm}
      \begin{align*}
        & \enc{\instance} = \enc{\instance,\mapping,\datadis; \work}\quad\text{where }\work =\prod_{\forall l_i\in L} \lconfig{l_i}{\emptyset}{\ex_l}\\
        &\enc{\instance,\mapping\cup (s,l),\datadis;\work\para \lconfig{l}{\emptyset}{\ex_l}} =
        \enc{\instance,\mapping,\datadis;\work\para \lconfig{l}{\emptyset}{\ex_l\para B_{l}(s)}}\\
        &\enc{\instance,\mapping,\datadis\cup \datadis(l);\work\para \lconfig{l}{\emptyset}{\ex_l}} =
        \enc{\instance,\mapping,\datadis;\work\para \lconfig{l}{\datadis(l)}{\ex_l}} \\
        &\enc{\instance,\emptyset,\emptyset;\work} =\work 
          \vspace*{-2mm}
      \end{align*}}
    \end{definition} 

    Formally, the encoding operator can be defined as a function with four input parameters: $(i)$ a workflow instance $\instance$ to be translated; $(ii)$ the set of pairs $(s,l)$ containing all mappings in $\mapping$; $(iii)$ the set $\datadis$ representing the distribution of the data over locations; $(iv)$ placeholder to build the workflow system $\work$. The translation starts by adding the auxiliary parameters into the encoding process and inside a SWIRL placeholder, creating a workflow containing locations configurations for each location in the workflow instance $\instance$ (for all $l\in L$). 
    The iteration process is divided in two phases, first at each iteration, the encoding function takes the pair $(s,l)$, identify the location $l$ and add the building block $B_{l}(s) $ into the execution trace to be executed on the location $l$. When all pairs are encoded, in the second phase, the iteration is on the distribution of the data over locations. Each set $\datadis(l)\subseteq \datadis$ is encoded to the corresponding locations.
    In that way, the trace $\ex_{l}$ is the parallel composition of building blocks $B_{l}(s)$ for each $s \in Q(l)$. The encoding finishes when both sets $\mapping$ and $\datadis$ are empty.
  \end{subsection}
  
  \begin{subsection}{Consistency of SWIRL semantics} \label{subsec:optimisation-consistency}
    This section defines a concurrency relation on the derivations of a workflow system $\work$, which is then used to show the consistency of different execution diagrams through the semantics. As commonly done in the literature, this section only considers reachable workflow systems defined below.

    \begin{definition}
      Given a distributed workflow instance $\instance =\left(W, L, \mathcal{M}, \Data{}, \mapdati\right)$ and the function $\enc{\cdot}: \mathcal{W}_{\instance} \lts{} \mathcal{W}_{\work} $, the \emph{initial state} of a distributed workflow system is 
      {\footnotesize
      \[
        \work_{Init}=\enc{\instance, L, *}= \prod_{l_j\in L} \lconfig{l_{j}}{\emptyset}{\prod_{s\in\load(l_j)} B(s)}
      \]}
    \end{definition}

    \begin{definition}
      A state of a workflow system $\work$ is \emph{reachable} if it can be derived from the initial state ($\work_{Init}$) by applying the rules in Figs.~\ref{fig:swirl-congruence-rules} and \ref{fig:swirl-semantics}.
    \end{definition}

    Having a transition $t:\work\lts{}\work'$, the workflow states $\work$ and $\work'$ are called \emph{source} and \emph{target} of the transition $t$, respectively. 
    The \emph{concurrency relation} is defined on the transitions having the same source. Formally:
    
    \begin{definition}[Concurrency relation]
      Two different transitions $t_1: \work \lts{}\work_1$ and $t_2: \work \lts{}\work_2$ having the same source, are always concurrent, written $t_1 \conc t_2$. 
    \end{definition}

    Following the standard notation, let $t_2/t_1$ represent a transition $t_2$ executed after the transition $t_1$. The concurrency relation is used to prove the \emph{Church-Rosser property}, which states that when two concurrent transitions execute at the same time, the ordering of the executions does not impact the eventual result. This finding shows that the concurrent semantics is confluent. Formally:

    \begin{lemma}[Church-Rosser property]\label{lm:church}
      Given two concurrent transitions $t_1: \work\lts{} \work_1$ and $t_2: \work\lts{} \work_2$, there exist two transitions $t_2/t_1: \work_1\lts{} \work_3$ and $t_1/t_2: \work_2\lts{} \work_3$ having the same target. 
    \end{lemma}
  \end{subsection}

\end{section}

\begin{section}{Optimisation} \label{sec:optimisation}
  This section introduces an \emph{optimisation function} that scans the entire workflow system to remove redundant communications, improving performance. In particular, there are two cases in which execution traces can be optimised: $(i)$ communications between steps deployed on the same location, which are always redundant; $(ii)$ multiple communications of the same data element between a pair of locations, when different steps mapped onto the destination location require the same input data from the same ports.
    
  The encoding function adds a communication to the workflow system $\work$ every time a data element is required for the execution of a step, no matter if it is already present at the destination location, creating unnecessary communications. For instance, consider a location $\lconfig{l}{D}{\ex}$ where $D = \emptyset$ and
  {\footnotesize
  \begin{equation*}
    \ex = \begin{aligned}[t]
      &\recv{p,l_1,l}.\exec{s,\inout{d}{d_1},\{l\}}.\send{\messa{d_1}{p_1},l,l}\para\\
      &\recv{p_1,l,l}.\exec{s_1, \inout{d_1}{\emptyset},\{l\}}
    \end{aligned}
  \end{equation*}}

  After the execution of the step $s$, the data element $d_1$ is saved on the location $l$ ($D\cup \{d_1\}$), therefore the $\mathtt{send}$/$\mathtt{recv}$ pair does not affect the state of $\work$. By removing the unnecessary communication, the trace $\ex$ can be rewritten as:

  {\footnotesize
  \begin{equation*}
    \ex' = \begin{aligned}[t]
      &\recv{p,l_1,l}.\exec{s,\inout{d}{d_1},\{l\}}\para\exec{s_1,\inout{d_1}{\emptyset}, \{l\}}
    \end{aligned}
  \end{equation*}}

  The rule (\textsc{Exec}) in Fig.~\ref{fig:swirl-semantics}, preserves dependency between steps $s$ and $s_1$ by ensuring that step $s_1$ will not execute until the required data $d_1$ is produced.

  The second optimisation step is to remove redundant communications between different pairs of locations when the same data element is sent multiple times through the same port. 
  For instance, consider two locations $\lconfig{l}{D}{\ex}$ and $\lconfig{l'}{D'}{\ex'}$ where $D = D' = \emptyset$ and 
  {\footnotesize
  \begin{align*}
    \ex  &= \recv{p,l_1,l}.\exec{s,\inout{d}{d_1},\{l\}}.(\prod^{3}_{i=1}\send{\messa{d_1}{p_1},l,l'})\\\vspace{-1mm}
    \ex' &= \prod^{3}_{i=1}\recv{p_1,l,l'}.\exec{s_i,\inout{d_1}{\emptyset},\{l'\}} 
    \vspace*{-2mm}
  \end{align*}}
  The first location $l$ sends the data element $d_{1}$ to three steps mapped onto location $l'$.
  Transferring the data element only once is enough, as the subsequent communications will not affect the state of $\work$, hence, there is:
  {\footnotesize
  \begin{align*}
    \ex  &= \recv{p,l_1,l}.\exec{s,\inout{d}{d_1},\{l\}}.\send{\messa{d_1}{p_1},l,l'}\\
    \ex' &=
      \begin{aligned}[t]
        &\recv{p_1,l,l'}.\exec{s_k,\inout{d_1}{\emptyset},\{l'\}}\para\prod^{3}_{i=1, i\neq k}\exec{s_i,\inout{d_1}{\emptyset},\{l'\}}
      \end{aligned}
  \end{align*}}

  The optimisation of a workflow system $\work$ is defined in terms of three functions: the first and the second\footnote{The two functions have the same notation to simplify the notation, they can be easily distinguished because of the different number of arguments.} ones start the optimisation process and controls it till the end, by taking the additional parameter $\prefset$ (the set of all prefixes/actions) and calling the third function that actually rewrite the execution trace of a location. It goes through the execution traces of the workflow $\work$ and breaks them into single action (prefix) blocks. Analysing the blocks one by one, it performs the following actions: $(i)$ if the predicate is a part of the communication on the same location, it is removed; $(ii)$ if the predicate is already in the set $\prefset$, it is removed as well (meaning the same data element was already sent to the same location through the same port, just to different step); $(iii)$ otherwise, the predicate is added to the set $\prefset$ and the drilling function moves to the next element.

  \begin{definition}\label{def:optimisationfun}
    Given the workflow system $\work$ and sets of workflow and optimised systems $\mathcal{W}_{\work}$ and $\mathcal{W}_{\opt}$, respectively, the optimisation function $\work$, $\localopt{\cdot}:\mathcal{W}_{\work}\lts{}\mathcal{W}_{\opt}$ is defined in terms of the auxiliary functions, $\localopt{\cdot}:\mathcal{W}_{\work}\times  \prefset\lts{}\mathcal{W}_{\opt}$ and $\traceopt{\cdot}:\mathcal{W}_{\work}\lts{}\mathcal{W}_{\opt}$ (where  $\operator \in\{|,.\}$ and $\prefset_{l,l}=\{\send{\messa{d}{p},l,l}, \recv{p,l,l}\}$) as follows: 
    
    {\footnotesize
    \vspace{-3mm}
    \begin{align*}
      \localopt{\work} &=\localopt{\traceopt{\work},\emptyset}  \qquad\qquad  \qquad\qquad\qquad \\
      \localopt{\traceopt{\lconfig{l}{\Data{}}{\ex}},\prefset} & = \lconfig{l}{\Data{}}{\localopt{\traceopt{\ex},\prefset}}
      \quad\quad \\
      \localopt{\traceopt{\work_1\para \work_2},\prefset}& = \localopt{\traceopt{\work_1},\prefset}\para \localopt{ \traceopt{\work_2},\prefset}\\
      \traceopt{\ex\;\operator \;\ex_1} &=\traceopt{\ex}\operator\traceopt{\ex_1} \qquad\qquad  \\
      \localopt{ \ex\;\operator \traceopt{\mu} \operator \traceopt{\ex_1},\prefset} & = 
      \begin{cases}
            \localopt{\ex\;\operator \;\nul\; \operator \traceopt{\ex_1},\prefset}
            \;\;\quad\;\;\;\text{ if } \mu\in \prefset \quad\vee
            \quad\mu\in\prefset_{l,l}\\
            \localopt{ \ex\;\operator \;\mu\; \operator \traceopt{\ex_1},\prefset\cup \mu} 
          \quad\text{ otherwise }
      \end{cases}\\
      \localopt{e, \prefset}& =e
    \end{align*}}
  \end{definition}

  The two workflow systems $\work$ and $\opt=\localopt{\work}$ are modelling the same behaviour of the distributed workflow system, i.e. the computations of the workflow steps are executed in the same order in both systems with the difference in the number of communications. Therefore, the weak barbed bisimulation~\cite{picalculus} is used to define the relation between the distributed workflow system and its optimised version. 
   
  To highlight that the executing action is a communication, it is labelled by $\tau$. Therefore, $\work\lts{\tau}\work'$ indicates that workflow system $\work$ can evolve into $\work'$ by performing the communication (transfer) action. The reflexive and transitive closure of $\lts{\tau}$ is denoted with $\xRightarrow{\text{$\tau$}}$ and the transition $\work\xRightarrow{\text{$\tau$}}\work'$ express the ability of the system $\work$ to evolve into $\work'$ by executing some number, possibly zero, of $\tau$ actions (communications). Given the transition $\work\lts{}\work'$ (any type of action, including the communication), if the same action can be executed after a certain number of communication actions, it is denoted as $\work\xRightarrow{\text{$\tau$}}\lts{}\work'$.

  The observable elements in this setting are the executions of the steps and it is denoted by $\work\barb{\nu}$ (resp. $\opt\barb{\nu}$) where $\nu=\exec{s, F(s), \mapping(s)}$ where the weak barb is denoted by $\work\wbarb{\nu}$ (resp. $\opt\wbarb{\nu}$), and it is defined as $\work\xRightarrow{\text{$\tau$}}\barb{\nu}$ (resp. $\opt\xRightarrow{\text{$\tau$}}\barb{\nu}$). Hence, the barbed bisimulation will check that all the step executions in a workflow system can be matched by the executions in the optimised one. 
    
  \begin{definition}\label{def:barbedbisimilarity}
    A relation $\rel\subseteq \mathcal{W}\times \mathcal{O}$ is a weak barbed simulation if $\work \rel\localopt{\work}$:
    \begin{itemize}
    \vspace{-1mm}
      \item $\work\barb{\nu}$ implies $\localopt{\work}\wbarb{\nu}$
      \item $\work\lts{}\work'$ implies $\localopt{\work}\Rightarrow \localopt{\work'}$ with $\work' \rel \localopt{\work'}$
    \end{itemize}
    A relation $\rel\subseteq \mathcal{W}\times \mathcal{O}$ is a weak barbed bisimulation if $\rel$ and $\rel^{-1}$ are weak barbed simulations. Weak bisimilarity, $\bisim$, is the largest weak barbed bisimulation.
  \end{definition}
       
  The next theorem shows the operational correspondence between a distributed workflow system and its optimised term.

  \begin{theorem}\label{thm:operational_correspondence}
    For any distributed workflow system $\work$, $\work\bisim\localopt{\work}$.
  \end{theorem}
\end{section}

\begin{section}{Implementation} \label{sc:implementation}
  The SWIRL compiler reference implementation\footnote{\url{https://github.com/alpha-unito/swirlc}}, called \texttt{swirlc}, follows the same line as the theoretical approach, separating scientific workflows' design and runtime phases. On the one hand, it allows the translation of high-level, product-specific workflow languages designed for direct human interaction to chains of low-level primitives easily understood by distributed runtime systems. A common representation fosters composability and interoperability among different workflow models, which can be easily combined into a single workflow system. Moreover, the translation process is performed with the formal consistency guarantees discussed in Sec.~\ref{subsec:encoding}.
  
  Finally, a SWIRL-based compiler can translate a workflow system $\work$ to a high-performance, self-contained workflow execution bundle based on send/receive communication protocols and runtime libraries, which can easily be included in a Research Object \cite{Bechhofer13}, improving reproducibility. An advanced compiler can also generate multiple execution bundles from the same workflow system, each optimised for a different execution environment (e.g., Cloud, HPC, or Edge), improving performance. As a bonus feature, the intrinsically distributed nature of SWIRL execution traces promotes decentralised runtime architectures, avoiding the single point of failure introduced by a centralised control plane.

  \begin{figure*}[t]
    \begin{center}
        \includegraphics[scale=0.6]{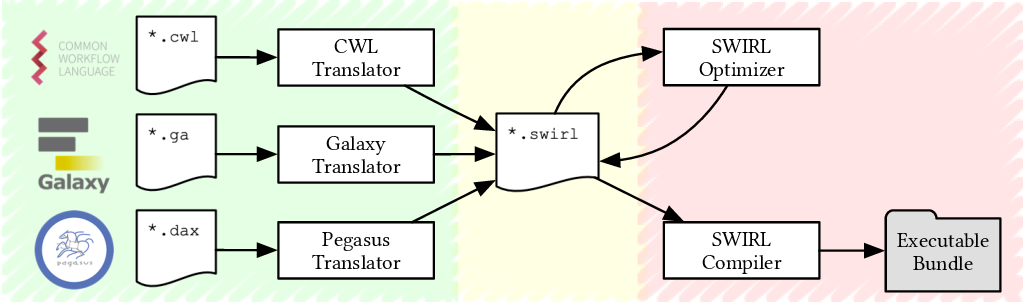}
        \caption{{\small{The SWIRL compiler toolchain.}}}
        \label{fig:swirl-toolchain}
    \end{center}
    \vspace{-6mm}
  \end{figure*}
  Fig.~\ref{fig:swirl-toolchain} sketches the SWIRL compiler toolchain. We implemented the SWIRL grammar using ANTLR \cite{Parr95}, and we automatically generated Python3 parser classes to process the SWIRL syntax. All the components of the SWIRL toolchain rely on these parsers to process $\mathtt{*.swirl}$ files. An abstract $\mathtt{SWIRLTranslator}$ class implements the encoding function, producing a SWIRL file from a workflow instance $\instance$. A concrete implementation specialises the $\mathtt{SWIRLTranslator}$ logic to the semantics of a given workflow language, e.g., CWL \cite{CWL:2022}, DAX (for Pegasus \cite{pegasus:2019}), or the Galaxy Workflow Format (GWF)\cite{galaxy:2022}\footnote{\label{ongoing}The implementation of the CWL and GWF translators and the $\mathtt{SWIRLOptimizer}$ are ongoing works.}. A $\mathtt{SWIRLOptimizer}$\footnotemark[\getrefnumber{ongoing}] class implements the optimisation function $\localopt{\cdot}:\mathcal{W}_{\work}\lts{}\mathcal{W}_{\opt}$, generating an optimised $\mathtt{*.swirl}$ file. Finally, an abstract $\mathtt{SWIRLCompiler}$ class produces an executable bundle from a $\mathtt{*.swirl}$ file and a declarative metadata file, which contains additional information not currently modelled in the SWIRL semantics, e.g., step commands, data types and location IP addresses. We have implemented a simple compiler class that generates a multithreaded Python program for each location, relying on TCP sockets for send/receive communications. 
\end{section}

\begin{section}{Evaluation} \label{sec:evaluation}
  \begin{SCfigure}[][t]
  \vspace{-2mm}
      \centering
        \includegraphics[scale=0.8]{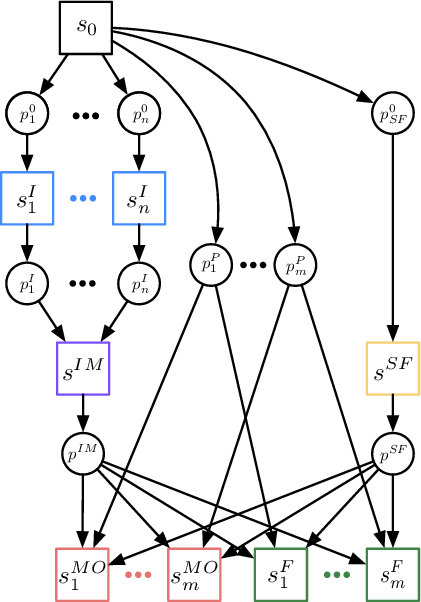}%
        \captionsetup{singlelinecheck=off}
        \caption[.]{
          \small{
              Graphical representation of the 1000 Genomes workflow contains five classes of steps mapped to diverse locations: $(i)$ \texttt{individuals} (blue), number of steps $n$, mapped to locations $l^{I}_{j}$, $j\in[1, a]$; $(ii)$ \texttt{individuals\_merge} (violet), a single step mapped to location $l^{IM}$; $(iii)$ \texttt{sifting} (yellow), a single step mapped to location $l^{IM}$; $(iv)$ \texttt{mutations\_overlap} (red), number of steps $m$, mapped to locations $l^{MO}_{t}$, $t\in[1,b],$ and $(v)$ \texttt{frequency} (green) number of steps $m$, mapped to locations $l^{F}_{k}$, $k\in [1, c]$. The initial step $s_{0}$, mapped to driver location $l_{d}$ is a step that sends each input data element to the correct location for processing.
              The mapping between data elements and ports, where $i\in[1, n]$ and $h \in [1, m]$, is:
              {\footnotesize
                  \begin{equation*}
                      \mapdati = \left\{
                          \begin{aligned}
                              & (d^{0}_i, p^{0}_i), (d^{P}_h, p^{P}_h), (d^{0}_{SF}, p^{0}_{SF}),\\
                              & (d^{I}_i, p^{I}_i),(d^{IM}, p^{IM}) , (d^{SF}, p^{SF}) 
                          \end{aligned}
                      \right\}
                  \end{equation*}
              }
              \vspace{-2mm}
          }
      }
      \label{fig:1000genome}
  \end{SCfigure}
    
  This section tests the flexibility of the SWIRL representation on the 1000 Genomes workflow \cite{Silva:19}, a Bioinformatics pipeline aiming at fetching, parsing and analysing data from the 1000 Genomes Project \cite{1000Genomes:15} to identify mutational overlaps and provide a null distribution for rigorous statistical evaluation of potential disease-related mutations. 
  However, we used the 1000 Genomes applications written in C++\cite{23:hipc:capio}.
  Fig.~\ref{fig:1000genome} shows a slightly simplified version of the 1000 Genomes workflow model. We removed some ports to simplify the notation, but their absence does not affect the reasoning reported in the rest of this Section. Note that the number of locations could be smaller than the number of steps. Hence, there could be a case when more steps are mapped to the same location.
  
  The corresponding workflow system $\work$ can be constructed using the encoding function $\enc{\cdot} : \mathcal{W}_{\instance} \lts{} \mathcal{W}_{\work}$ (Sec.~\ref{subsec:encoding}). It can be written as follows:
  {\footnotesize
  \begin{align*}
    \work = &\prod_{i\in\{d,SF,IM\}}\lconfig{l^{i}}{\emptyset}{\ex^{i}}
    \para \prod_{j=1}^{a} \lconfig{l^{I}_{j}}{\emptyset}{\ex^{I}_{j}}
    \para  \prod_{t=1}^{b} \lconfig{l^{MO}_{t}}{\emptyset}{\ex^{MO}_{t}}
    \para \prod_{k=1}^{c} \lconfig{l^{F}_{k}}{\emptyset}{\ex^{F}_{k}}
  \end{align*}}
  where each execution trace $\ex^{*}_{*}$ defines the actions (steps and data transfers) depicted in Fig.~\ref{fig:1000genome} to be executed on the corresponding location. For instance, if the driver location is taken, the the execution trace $e^{d}$ is defines as:
    {\footnotesize
  \begin{align*}
    \ex^{d}= &\prod_{i=1}^{n} \send{\messa{d^{0}_i}{p^{0}_i},l^{d},l^{I}_{j}}\para \send{\messa{d^{0}_{SF}}{p^{0}_{SF}},l^{d},l^{SF}}\para\\
    &\prod_{h=1}^{m} (\send{\messa{d^{P}_h}{p^{P}_h},l^{d},l^{MO}_{t}}\para \send{\messa{d^{P}_h}{p^{P}_h},l^{d},l^{F}_{k}})
  \end{align*}}
  The full representation of $\work$ is discussed in Appendix~\ref{appendix:evaluation}. The 1000 Genomes workflow modelled above can be reproduced using the SWIRL implementation (Section \ref{sc:implementation}). To keep the experiment small and ease reproducibility, the ten homogeneous execution locations and a single chromosome, i.e., a single workflow instance, are considered. The necessary installing package and instructions on how to run the experiment can be find in \cite{zenodo-swirl}.
\end{section}

\begin{section}{Conclusion} \label{sc:conclusion}
  This work introduced SWIRL, a ``Scientific Workflow Intermediate Representation Language'' based on send/receive communication primitives. An encoding function maps any workflow instance onto a distributed execution plan $\work$, fostering interoperability and composability of different workflow models. A set of rewriting rules allows for automatic optimisation of data communications, improving performance with correctness and consistency guarantees. The optimised SWIRL representation can be compiled into one or more self-contained executable bundles addressing specific execution environments, ensuring reproducibility and embracing heterogeneity. SWIRL already proved itself to be flexible enough to model a real large-scale scientific workflow (even if still not supporting all features of modern WMSs). 

  The foundational contribution of SWIRL is to propose a novel direction to solve well-known problems in the field of scientific workflows. Indeed, SWIRL shifts the focus from high-level workflow languages, designed either for direct human interaction or to encode complex, product-specific features, to a low-level minimalistic representation of a workflow execution plan, which is far more manageable from both formalisation methods and compiler toolchains. In this context, we hope that SWIRL can pave the way to a novel, more formal approach to distributed workflow orchestration research.

  The formal SWIRL representation gives the possibility to build the formally correct extensions, for instance, adding a type system where the multiparty sessions are enriched with security levels for messages (data in our case)\cite{Capecchi:16} or deriving the causal-consistent reversible framework by applying the approach\cite{Lanese:20}, that later can be used as a base to build fault-tolerance mechanism.
\end{section}

\begin{credits}
  \subsubsection{\ackname}
  This work was supported by: the Spoke 1 ``FutureHPC \& BigData'' of ICSC - Centro Nazionale di Ricerca in High-Performance Computing, Big Data and Quantum Computing, funded by European Union - NextGenerationEU; the EUPEX EU’s Horizon 2020 JTI-EuroHPC research and innovation programme project under grant agreement No 101033975.
\end{credits}

\clearpage


\appendix

\setcounter{page}{1}

\begin{section}{Proofs}\label{app:a}
  Having a transition $t:\work\lts{}\work'$, the workflow states $\work$ and $\work'$ are called \emph{source} and \emph{target} of the transition $t$, respectively. Two transitions are defined \emph{coinitial} if they have the same source and \emph{cofinal} if they have the same target. When the target of a transition $t_{1}$ is the source of the transition $t_{2}$, the two transitions are said to be \emph{composable}, and their composition is denoted with $t_1;t_2$. A \emph{derivation} is a sequence of composable transitions, and it is written as $\mathtt{t}:\work\lts{\;}^{*}\work'$, while the special symbol $\varepsilon$ denotes the empty derivation.

  \begin{lemma}[Church-Rosser property]\label{lm:church-appendix}
    Given two concurrent transitions $t_1: \work\lts{} \work_1$ and $t_2: \work\lts{} \work_2$, there exist two cofinal transitions $t_2/t_1: \work_1\lts{} \work_3$ and $t_1/t_2: \work_2\lts{} \work_3$. 
  \end{lemma}
  
  \begin{proof} 
    Having two concurrent transitions $t_1: \work\lts{\mu_1} \work_1$ and $t_2: \work\lts{\mu_2} \work_2$, there are two cases:
    \begin{center}
      \begin{tikzpicture}
        [scale=.4,auto=left,every node/.style={rectangle}]
        \node (n1) at (0,3.5) {$\work_1$};
        \node (n2) at (3.5,0)  {$\work_3$};
        \node (n3) at (3.5,7)  {$ \work$};
        \node (n4) at (7,3.5) {$\work_2$};
          \draw[->, thick]  (n3) -- (n1) node[pos=.6, above=1mm] {$t_1$} ; 
          \draw[->, thick]  (n3) --  (n4) node[pos=.6, above=1mm] {$t_2$}; 
          \draw[->, dotted]  (n1) --  (n2) node[pos=.1, below=2mm] {$t_2/t_1$}; 
          \draw[->, dotted]  (n4) --  (n2) node[pos=.1, below=2mm] {$t_1/t_2$}; 
      \end{tikzpicture}
    \end{center}
    \begin{itemize}
      \item when transitions $t_1$ and $t_2$ are executed on different locations; in this case the system $\work$ can be represented as $\work_1\para \work_2$ where transition $t_1$ ($t_2$) is executed on the part of the workflow $\work_1$ ($\work_2$). Since transitions are concurrent they do not affect each other, more precise, transition $t_1$ does not change the part of the workflow system on which the transition $t_2$ is executing and viceversa. Therefore, by executing $t_1$ the obtained system is $\work'=\work'_1\para\work_2$, and by executing now $t_2/t_1$ the result is $\work''=\work'_1\para\work'_2$. Similarly, for the execution of the transitions $t_2$ and $t_1/t_2$, the obtained systems are $\work'''=\work_1\para\work'_2$ and $\work''=\work'_1\para\work'_2$ respectively.
      \item when $t_1$ and $t_2$ are executed on the same location; since transitions are concurrent and different, system $\work$ can be represented as $\lconfig{l}{\Data{} }{\ex_1\para\ex_2}\para \work'$ and the transition $t_1$ executes action on the left part of the execution trace, i.e. $\ex_1$, while $t_2$ involves trace $\ex_2$. Again two transitions do not affect each other, since transitions are concurrent, and the reasoning is similar as in the first case with the fact that transitions are executed on the same location.
    \end{itemize}
  \end{proof}

  \begin{lemma}[Correspondence]\label{lm:correspondence-appendix}
    For all actions $\work\lts{}\work'$ and $\opt=\localopt{\work}$, there is a corresponding action $\opt\weak \opt'$  with $\opt'= \localopt{\work'}$.
  \end{lemma}
     
  \begin{proof} By induction on the derivation of $\work\lts{}\work'$ with the case analysis on  the last applied rule:
    \begin{itemize}
      \item {\bf{(\textsc{Exec})}} then the workflow $\work$ has the form: 
      {\small
      \[
      \work=\prod_{i, l_i \in\mapping(s)}\lconfig{l_i}{\Data{i} }{\exec{s,F(s),\mapping(s)}.\ex_i}
      \]}
      If the data elements necessary for the execution of the step $s$ are on the corresponding locations (i.e. $\In^{D}(s)\subseteq \Data{i}$), the rule (\textsc{Exec}) can be applied and the obtained workflow is:
      {\small
      \[
      \work'=\prod_{i, l_i \in\mapping(s)}\lconfig{l_i}{\Data{i} \cup \Out^{D}(s) }{\ex_i}
      \]}
      The optimised workflow $\opt=\localopt{\work}$ is:
      {\small
      $$\opt=\prod_{i, l_i \in\mapping(s)}\lconfig{l_i}{\Data{i} }{\exec{s,F(s),\mapping(s)}.\localopt{\ex_i,A_i}}$$ }
      with the necessary data elements on the corresponding locations. On the optimised workflow, the rule (\textsc{Exec}) can be applied and there is:
      {\small
      $$\opt'=\prod_{i, l_i \in\mapping(s)}\lconfig{l_i}{\Data{l} \cup \Out^{D}(s)}{\localopt{\ex_i,A_i}}$$ }
      with $\opt'=\localopt{\work'}$. Given that the executed action was $\exec{}$, from $\opt\lts{}\opt'$ there is $\opt\weak\opt'$ where no $\tau$ action is performed.
      \item {\bf{(\textsc{L-Comm})}} in this case, the workflow $\work$ has the form of: 
      {\small
      \[\work=\lconfig{l}{\Data{} }{\send{\messa{d}{p},l,l }.\ex\para \recv{p,l,l }.\ex'}\]}
      If $d\in \Data{}$, the rule (\textsc{L-Comm}) can be applied and the obtained workflow system is:
      {\small
      $$\work'=\lconfig{l}{\Data{} }{\ex\para\ex'}$$}
      The optimised workflow system $\opt=\localopt{\work}$ is:
      {\small
      \begin{align*}
      \opt&=\lconfig{l}{\Data{} }{\localopt{\send{\messa{d}{p},l,l }.\ex\para \recv{p,l,l }.\ex',\prefset}}\\
      &=\lconfig{l}{\Data{} }{\localopt{\ex\para \ex',\prefset}}
      \end{align*}}
      By applying the optimisation function on $\work'$ there is:
      {\small
      \[
        \localopt{\work'}=\lconfig{l}{\Data{} }{\localopt{\ex\para \ex',\prefset}}
      \]}
      Since the workflow $\work$ performs communication ($\tau$), by definition of $\weak$ (in this case $\work_1\weak\work_2$ means that the workflow $\work_1$ evolves into $\work_2$ by applying some number of communications, could be zero), there is:
      {\small
      $$\lconfig{l}{\Data{} }{\localopt{\ex\para\ex',\prefset}}\weak \lconfig{l}{\Data{} }{\localopt{\ex\para\ex',\prefset}}=\opt'$$}
      \item {\bf{(\textsc{Comm})}} then the workflow $\work$ has the form:
      {\small
      $$\work=\lconfig{l}{\Data{}}{\send{\messa{d}{p},l,l' }.\ex} \para  \lconfig{l'}{D'}{\recv{p,l,l' }.\ex'}$$}
      If $d\in \Data{}$, then the rule (\textsc{Comm}) can be applied and the obtained workflow system is:
      {\small
      $$\work'=\lconfig{l}{\Data{}}{\ex}\para \lconfig{l'}{D'\cup \{d\}}{\ex'}$$}
      For the optimised workflow $\opt=\localopt{\work}$ there are two possibilities depending on whether the communication between locations $l$ and $l'$ is the first one that transfers the data element $d$ or not:
      {\bf{(i)}} if the communication is the first one, then it is not removed by the optimising function and there is:
      {\small
      \begin{align*}
      &\opt=\lconfig{l}{\Data{}}{\send{\messa{d}{p},l,l' }.\localopt{\ex, \prefset\cup \send{p,l,l' }}} \para \\
      &\lconfig{l'}{D'}{\recv{p,l,l' }.\localopt{\ex', \prefset'\cup \recv{\messa{x}{p},l,l' }}}
      \end{align*}}
      and the rule (\textsc{Comm}) can be applied resulting in:
      {\small
      $$\opt'=\lconfig{l}{\Data{}}{\localopt{\ex,\prefset_1}} \para  \lconfig{l'}{D'\cup \{d\}}{\localopt{\ex',\prefset'_1}}$$}
      where $\prefset_1=\prefset\cup \send{p,l,l' }$, similar for $\prefset'_1$. Now by applying optimising function on $\work'$, there is:
      {\small
      $$\localopt{\work'}=\lconfig{l}{\Data{}}{\localopt{\ex,\prefset_1}} \para  \lconfig{l'}{D'\cup \{d\}}{\localopt{\ex',\prefset'_1}}=\opt'$$}
      {\bf{(ii)}} if the executing communication is not the first one, then it is removed by the optimising function and there is:
      {\small
      $$\opt=\lconfig{l}{\Data{}}{\localopt{\ex,\prefset}} \para  \lconfig{l'}{D'}{\localopt{\ex',\prefset'}}$$}
      By definition of $\weak$ it is possible to write $\opt\weak\opt'$ where $\weak$ represents zero communication actions (therefore $\opt=\opt'$). Then, there is:
      {\small
      $$ \lconfig{l}{\Data{}}{\localopt{\ex,\prefset}} \para  \lconfig{l'}{D'}{\localopt{\ex',\prefset}}=\opt'$$}
      and $\localopt{\work'}=\opt$.
      \item {\bf{(\textsc{L-Par})}}
      then the workflow $\work$ has the form: 
      {\small
      $$\work=\lconfig{l}{\Data{} }{\ex_1\para\ex_2}$$}
      With the premise that the action $\lconfig{l}{\Data{}}{\ex_1 }\lts{}\lconfig{l}{D'}{\ex'_1 }$ can be executed, on the workflow system $\work$, the rule (\textsc{L-Par}) can be applied and the obtained term is:
      {\small
      $$\work'=\lconfig{l}{D' }{\ex'_1\para\ex_2}$$}
      The optimised workflow $\opt=\localopt{\work}$ is:
      {\small
      $$\opt=\lconfig{l}{\Data{} }{\localopt{\ex_1\para\ex_2,A}}$$}
      By inductive hypothesis, if $\lconfig{l}{\Data{}}{\ex_1 }\lts{}\lconfig{l}{D'}{\ex'_1 }$ executed then $\opt_1\weak\opt'_1$ with $\opt_1=\localopt{\lconfig{l}{\Data{}}{\ex_1 }}=\lconfig{l}{D_{}}{\localopt{\ex_1,A}}$ and $\opt'_1=\lconfig{l}{D'}{\localopt{\ex'_1,A'}}$.
        With the hypothesis that $\opt_1\weak\opt'_1$, $\opt\weak\opt'$ where $\opt'=\localopt{\work'}$.
      \item {\bf{(\textsc{Sec})}}
      then the workflow $\work$ has the form: 
      {\small
      $$\work=\lconfig{l}{\Data{} }{\ex_1.\ex_2}$$}
      With the hypothesis that the action $\lconfig{l}{\Data{}}{\ex_1 }\lts{}\lconfig{l}{D'}{\ex'_1 }$ can be executed, on the workflow system $\work$, the rule (\textsc{Sec}) can be applied and the obtained term is:
      {\small
      $$\work'=\lconfig{l}{D' }{\ex'_1.\ex_2}$$}
      The optimised workflow $\opt=\localopt{\work}$ is:
      {\small
      $$\opt=\lconfig{l}{\Data{} }{\localopt{\ex_1.\ex_2,A}}$$}
      By inductive hypothesis, if $\lconfig{l}{\Data{}}{\ex_1 }\lts{}\lconfig{l}{D'}{\ex'_1 }$ executed then $\opt_1\weak\opt'_1$ with $\opt_1=\localopt{\lconfig{l}{\Data{}}{\ex_1 }}=\lconfig{l}{\Data{}}{\localopt{\ex'_1,A}}$ and $\opt'_1=\lconfig{l}{D'}{\localopt{\ex'_1,A'}}$.
        With the hypothesis that $\opt_1\weak\opt'_1$, there is $\opt\weak\opt'$ where $\opt'=\localopt{\work'}$.
      \item {\bf{(\textsc{Par})}}
      then the workflow $\work$ has the form: 
      {\small
      $$\work=\work_1\para\work_2$$}
      With the hypothesis that the action $\work_1\lts{}\work'_1$ can be executed, on the workflow system $\work$, the rule (\textsc{Par}) can be applied and the obtained term is:
      {\small
      $$\work'=\work'_1\para\work_2$$}
      The optimised workflow $\opt=\localopt{\work}$ is:
      {\small
      $$\opt=\localopt{\work_1,A}\para\localopt{\work_2,A}$$}
      By inductive hypothesis, if $\work_1\lts{}\work'_1$ executed then $\opt_1\weak\opt'_1$ with $\opt_1=\localopt{\work_1}=\localopt{\work_1,\prefset}$ and $\opt'_1=\localopt{\work'_1,A'}$.
      With the hypothesis that $\opt_1\weak\opt'_1$, there is $\opt\weak\opt'$ where $\opt'=\localopt{\work'}$.
      \item {\bf{(\textsc{Congr})}}
      Given the workflow $\work$, the executed transition is $\work\lts{}\work'$ with the premise $\work\congr\work_1 \lts{}\work'_1 \congr \work'$.
      By induction hypothesis, there is the optimised system $\opt_1=\localopt{\work_1}$ such that $\opt_1\weak\opt'_1$ with $\opt'_1=\localopt{\work'_1}$.
      For $\opt=\localopt{\work}$ by the definition of the optimising function, $\work\congr\work_1$ implies $\opt=\localopt{\work}\congr\localopt{\work_1}=\opt_1$. Then the rule (\textsc{Congr}) can be applied and $\opt\weak\opt'$ where $\opt'=\localopt{\work'}$.
    \end{itemize}
  \end{proof}

  \begin{lemma}[Completeness] \label{lm:completeness-appendix}
    For workflow system $\work$, such that $\opt=\localopt{\work}$ if $\opt\lts{} \opt'$ then there exists the corresponding action $\work\weak\work'$ with $\opt'= \localopt{\work'}$. 
  \end{lemma}
  \begin{proof}
    By induction on the derivation of $\opt\lts{} \opt'$ where $\opt=\localopt{\work}$ with the case analysis on the last applied rule.    The proof is similar to the one of Lemma~\ref{lm:correspondence-appendix}.
  \end{proof}

  \begin{theorem}[Operational Correspondance]\label{thm:operational_correspondence-appendix}
    For any distributed workflow system $\work$, $\work\bisim\localopt{\work}$.
  \end{theorem}
  \begin{proof} It is necessary to show that the relation 
    \[
      \rel=\{(\work,\localopt{\work}) \text{ where }  \work \text{ is reachable }\}
    \]
      is weak barb bisimulation. 
  
    For the first condition of Definition~\ref{def:barbedbisimilarity}, if $\work$ has a barb $\nu$, then the optimised term $\localopt{\work}$ has the same barb, since the optimising function does not interfere with the predicates $\exec{\cdot}$. The same reasoning can be applied on the optimised term $\localopt{\work}$, if it has barb $\nu$, then the workflow $\work$ has it as well.
  
    If workflow $\work$ can perform the transition $\work\lts{}\work'$ and $\opt=\localopt{\work}$, then by Lemma~\ref{lm:correspondence-appendix}, there is $\opt\weak\opt'$ with $\opt'=\localopt{\work'}$ ($(\work',\localopt{\work'})\in\rel$).
  
    In the other direction, if for a given workflow system $\work$, $\opt=\localopt{\work}$ and $\opt\lts{}\opt'$, then by Lemma~\ref{lm:completeness-appendix}, there is $\work\weak\work'$ with $\opt'=\localopt{\work'}$ ($(\work',\localopt{\work'})\in\rel$).  
  \end{proof}
\end{section}

\begin{section}{Evaluation} \label{appendix:evaluation}
  \begin{figure}[t]
    \begin{center}
        \includegraphics[scale=0.85]{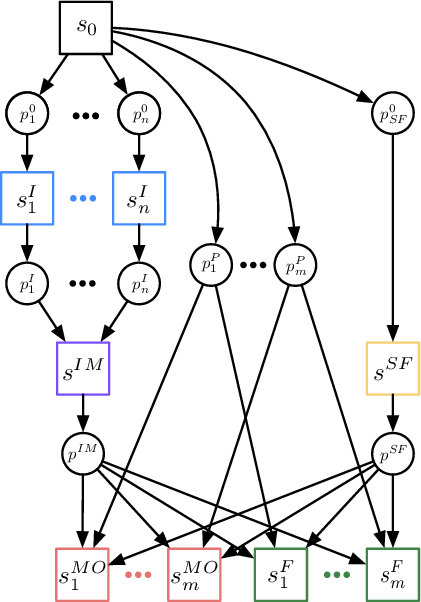}
        \caption{Graphical representation of the 1000 Genomes workflow. The workflow contains five classes of steps: individuals (blue), individuals\_merge (violet), sifting (yellow), mutations\_overlap (red), and frequency (green). The initial step $s_{0}$ is an auxiliary step that sends each input data element to the correct location for processing.}
        \label{fig:1000genome-app}
    \end{center}
    \vspace{-7mm}
  \end{figure}
  
  \begin{table}[t]
    \begin{center}
      \begin{tabular}{l|c|c|c}
        \toprule
        Step name & \# steps & \# locations & Location name  \\ 
        \midrule
        $s_{0}$ & 1 & 1 & $l^{d}$\\[1mm]
        \texttt{individuals} & $n$ & $a$ & $l^{I}_{j}$, $j=1,\ldots,a$ \\[1mm]
        \texttt{individuals\_merge} & 1 & 1 & $l^{IM}$\\ [1mm]
        \texttt{sifting} & 1 & 1 & $l^{SF}$ \\ [1mm]
        \texttt{mutations\_overlap} & $m$ & $b$ & $l^{MO}_{t}$, $t=1,\ldots,b$\\ [1mm]
        \texttt{frequency} & $m$ & $c$ & $l^{F}_{k}$, $k=1,\ldots,c$ \\
        \bottomrule
      \end{tabular}
    \end{center}
    \caption{Characterisation of the 1000 Genomes workflow instance to be encoded in a SWIRL workflow system.}
    \label{tab:1000genome-locations}
  \end{table}
    
  This section tests the flexibility of the SWIRL representation on the 1000 Genomes workflow \cite{Silva:19}, a Bioinformatics pipeline aiming at fetching, parsing and analysing data from the 1000 Genomes Project to identify mutational overlaps and provide a null distribution for rigorous statistical evaluation of potential disease-related mutations. Fig.~\ref{fig:1000genome-app} shows a slightly simplified version of the 1000 Genomes workflow model. We removed some ports to simplify the notation, but their absence does not affect the reasoning reported in the rest of this section. Table~\ref{tab:1000genome-locations} reports the number of steps and mapping locations for each class of steps in a distributed instance of the 1000 Genomes workflow. The mapping between data elements and ports is written as follows:

  {\footnotesize
  \vspace{2mm}
  \[
    \mapdati= \left\{  \begin{aligned}  & (d^{0}_i, p^{0}_i), (d^{P}_h, p^{P}_h), (d^{0}_{SF}, p^{0}_{SF}),\\
      & (d^{I}_i, p^{I}_i),(d^{IM}, p^{IM}) , (d^{SF}, p^{SF}) \end{aligned} \right\}
  \]}
  where $i\in[1, n]$ and $h \in [1, m]$.
    
  The corresponding workflow system $\work$ can be constructed using the encoding function $\enc{\cdot} : \mathcal{W}_{\instance} \lts{} \mathcal{W}_{\work}$ (Sec.~\ref{subsec:encoding}). We did not explicitly consider the case of ports without input steps at the beginning of a workflow model. However, this limitation can be easily overcome by introducing an auxiliary initial step $s_{0}$ with an empty execution trace $\nul$, whose only purpose is distributing initial data from the local location $l_{d}$ to the target locations for further processing. The resulting workflow system is then written as follows:
    
  {\footnotesize
  \begin{align*}
    \work = &\lconfig{l^{d}}{\emptyset}{\ex^{d}}\para \lconfig{l^{SF}}{\emptyset}{\ex^{SF}}\para \prod_{j=1}^{a} \lconfig{l^{I}_{j}}{\emptyset}{\ex^{I}_{j}}\para \\
            &\lconfig{l^{IM}}{\emptyset}{\ex^{IM}}\para \prod_{t=1}^{b} \lconfig{l^{MO}_{t}}{\emptyset}{\ex^{MO}_{t}}\para \prod_{k=1}^{c} \lconfig{l^{F}_{k}}{\emptyset}{\ex^{F}_{k}}
  \end{align*}}
  where the execution trace $e^{d}$ iterates over locations $l^{I}_{j}\in\mapping(s^{I}_i)$, $l^{MO}_{t}\in\mapping(s^{MO}_h)$ and $l^{F}_{k}\in\mapping(s^{F}_h)$ as follows:

  {\footnotesize
  \begin{align*}
    \ex^{d}= &\prod_{i=1}^{n} \send{\messa{d^{0}_i}{p^{0}_i},l^{d},l^{I}_{j}}\para \send{\messa{d^{0}_{SF}}{p^{0}_{SF}},l^{d},l^{SF}}\para\\
    &\prod_{h=1}^{m} (\send{\messa{d^{P}_h}{p^{P}_h},l^{d},l^{MO}_{t}}\para \send{\messa{d^{P}_h}{p^{P}_h},l^{d},l^{F}_{k}})
  \end{align*}}
  Note that the number of locations could be smaller than the number of steps. Hence, there could be a case when more steps are mapped to the same location.
    
  The behaviour of each location $l^{I}_{j}$ is to collect all the steps $s^{I}_i$ mapped to it, receive input data from them, execute the computation of step $s^{I}_{j}$ and send the result to location $l^{IM}$ to perform the individuals\_merge step:
    
  {\footnotesize
  \begin{align*}
    \ex^{I}_{j} = & \prod_{s^{I}_i\in Q(l^{I}_{j})} \recv{p^{0}_i,l^{d},l^{I}_{j}}.\exec{s^{I}_i,\inout{d^{0}_i}{d^{I}_{i} },\{l^{I}_{j}\}}.\\
    &\send{\messa{d^{I}_{i}}{p^{I}_{i}},l^{I}_{j},l^{IM}}
  \end{align*}}
  Location $l^{IM}$ receives all the data elements sent by the individuals, computes the individuals\_merge step and sends the produced data to all mutations\_overlap and frequency steps with the corresponding locations. Its execution trace is:
    
  {\footnotesize
    \begin{equation*}
      \ex^{IM} = \begin{aligned}[t]
        &\prod_{i=1}^{n}\recv{p^{I}_{i}, l^{I}_{j},l^{IM}}.\\
        &\exec{s^{IM},\inout{d^{I}_{1},\ldots, d^{I}_{n}}{d^{IM}},\{l^{IM}\}}.\\
        &\big(\prod_{t=1}^{b}\prod^{|\load(l^{MO}_{t})|} \send{\messa{d^{IM}}{p^{IM}},l^{IM},l^{MO}_{t}}\para \\
        &\prod_{k=1}^{c}\prod^{|\load(l^{F}_{k})|} \send{\messa{d^{IM}}{p^{IM}},l^{IM},l^{F}_{k}}\big)
      \end{aligned}
    \end{equation*}
  }

  The behaviour of location $l^{SF}$, which executes the sifting step, is analogous to the one of $l^{IM}$. Note that all these communication schemes would greatly benefit from an optimised $\mathtt{broadcast}$ operator. Adding more advanced communication primitives and their  composition to SWIRL is planned. 
    
  Steps $s^{MO}_{t}$ and $s^{F}_{t}$ are equivalent in terms of communication patterns, so we just consider the first one. Each location $l^{MO}_{t}$ receives data from locations $l^{IM}$, $l^{SF}$, and $l^{d}$ and executes step $s^{MO}_{i}$. The resulting execution trace $\ex^{MO}_{t}$ is written as follows:  
    
  {\footnotesize
    \begin{align*}
      \ex^{MO}_{t} = 
        & \prod_{s^{MO}_h\in Q(l^{MO}_{t})}
        \begin{aligned}[t]\big(
          &\recv{p^{IM}, l^{IM}, l^{MO}_{t}}\para\\
          &\recv{p^{SF}, l^{SF}, l^{MO}_{t}}\para\\
          &\recv{p^{P}_h, l^{d}, l^{MO}_{t}}
        \big).\end{aligned}\\
        &\exec{s^{MO}_{h},\inout{d^{SF}, d^{IM},d^{P}_h }{\emptyset },\{l^{MO}_{t}\}}
    \end{align*}
  }

  When $m > b$, execution traces $e^{IM}$ and $e^{MO}_{t}$ can be optimised to communicate each data element just once between locations $l^{IM}$ and $l^{MO}_{t}$. The optimised traces $o^{IM}$ and $o^{MO}_{t}$, computed by the optimisation function $\localopt{\work}$, can be written as follows:

  {\footnotesize
    \begin{align*}
      o^{IM} &= 
        \begin{aligned}[t]
          &\prod_{i=1}^{n}\recv{p^{I}_{i}, l^{I}_{j},l^{IM}}.\\
          &\exec{s^{IM},\inout{d^{I}_{1},\ldots, d^{I}_{n}}{d^{IM}},\{l^{IM}\}}.\\
          &\big(\prod_{t=1}^{b} \send{\messa{d^{IM}}{p^{IM}},l^{IM},l^{MO}_{t}}\para \\
          &\prod_{k=1}^{c}\prod^{|\load(l^{F}_{k})|} \send{\messa{d^{IM}}{p^{IM}},l^{IM},l^{F}_{k}}\big)
        \end{aligned}\\[3mm]
      o^{MO}_{t} &=
        \begin{aligned}[t]
          &\recv{p^{IM}, l^{IM}, l^{MO}_{t}}\para\recv{p^{SF}, l^{SF}, l^{MO}_{t}}\para\\
          &\prod_{s^{MO}_h\in Q(l^{MO}_{t})} \recv{p^{P}_h, l^{d}, l^{MO}_{t}}.\\
          &\exec{s^{MO}_{h},\inout{d^{SF}, d^{IM},d^{P}_h }{\emptyset },\{l^{MO}_{t}\}}
        \end{aligned}
    \end{align*}
  }

  With analogous reasoning, it is possible to apply further optimisations whenever $m > c$ and $n > a$, significantly reducing data transfer overhead.
\end{section}
\end{document}